\documentclass[a4paper,UKenglish,cleveref, autoref, thm-restate]{lipics-v2021}

\pdfoutput=1 
\hideLIPIcs  

\title{Single Family Algebra Operation on BDDs and ZDDs Leads To Exponential Blow-Up} 

\author{Kengo Nakamura}{NTT Communication Science Laboratories, Kyoto, Japan}{kengo.nakamura@ntt.com}{https://orcid.org/0000-0002-9615-3479}{}

\author{Masaaki Nishino}{NTT Communication Science Laboratories, Kyoto, Japan}{masaaki.nishino@ntt.com}{https://orcid.org/0000-0001-6489-5446}{}
\author{Shuhei Denzumi}{NTT Communication Science Laboratories, Kyoto, Japan}{shuhei.denzumi@ntt.com}{https://orcid.org/0000-0002-0794-4157}{was supported by JSPS KAKENHI Grant Number JP23H04391.}

\authorrunning{K. Nakamura, M. Nishino, and S. Denzumi} 

\Copyright{Kengo Nakamura, Masaaki Nishino, and Shuhei Denzumi} 



\category{} 

\relatedversiondetails{Full Version}{https://arxiv.org/abs/2403.05074} 

\funding{This work was generally supported by JSPS KAKENHI Grant Number JP20H05963 and JST CREST Grant Number JPMJCR22D3.}

\acknowledgements{We thank Hiromi Emoto and Shou Ooba for pointing out the issues regarding the complexity of performing family algebra operations on ZDDs.
We also thank Shin-ichi Minato, Jun Kawahara, and Norihito Yasuda for valuable discussions on this topic.
I am grateful to the reviewers of ISAAC for the comments to improve the manuscript.}

\nolinenumbers 

\EventEditors{Juli\'{a}n Mestre and Anthony Wirth}
\EventNoEds{2}
\EventLongTitle{35th International Symposium on Algorithms and Computation (ISAAC 2024)}
\EventShortTitle{ISAAC 2024}
\EventAcronym{ISAAC}
\EventYear{2024}
\EventDate{December 8--11, 2024}
\EventLocation{Sydney, Australia}
\EventLogo{}
\SeriesVolume{322}
\ArticleNo{19}

\usepackage{mathtools}
\usepackage{algorithmic}
\usepackage{graphicx}
\usepackage{textcomp}
\usepackage{enumerate}
\usepackage{multirow}
\usepackage{url}
\usepackage{pifont}
\usepackage{textcomp}
\usepackage{booktabs}
\usepackage{etoolbox}

\newcommand{\calC}{\mathcal{C}}

\newcommand{\calE}{\mathcal{E}}
\newcommand{\calF}{\mathcal{F}}
\newcommand{\calG}{\mathcal{G}}
\newcommand{\calH}{\mathcal{H}}

\newcommand{\calP}{\mathcal{P}}
\newcommand{\calQ}{\mathcal{Q}}

\newcommand{\calT}{\mathcal{T}}

\newcommand{\ZDDZ}{\mathtt{Z}}
\newcommand{\ZDDA}{\mathtt{A}}
\newcommand{\ZDDN}{\mathtt{N}}
\newcommand{\ZDDr}{\mathtt{r}}
\newcommand{\ZDDn}{\mathtt{n}}
\newcommand{\ZDDm}{\mathtt{m}}
\newcommand{\ZDDlo}{\mathsf{lo}}
\newcommand{\ZDDhi}{\mathsf{hi}}
\newcommand{\ZDDlabel}{\mathsf{lb}}

\DeclareMathOperator{\ddelta}{\boxplus}
\DeclareMathOperator{\djoin}{\dot{\bowtie}}
\DeclareMathOperator{\jjoin}{\hat{\bowtie}}
\DeclareMathOperator{\sdiv}{/}
\DeclareMathOperator{\modulus}{\%}
\DeclareMathOperator{\nonsupset}{\searrow}
\DeclareMathOperator{\nonsubset}{\nearrow}

\newcommand{\order}[1]{O( #1 )}
\newcommand{\morder}[1]{\Omega( #1 )}
\newcommand{\sorder}[1]{o( #1 )}
\newcommand{\poly}{\mathrm{poly}}
\newcommand{\comp}{\prime}

\begin{document}

\maketitle

\begin{abstract}
  Binary decision diagram (BDD) and zero-suppressed binary decision diagram (ZDD) are data structures to represent a family of (sub)sets compactly, and it can be used as succinct indexes for a family of sets.
  To build BDD/ZDD representing a desired family of sets, there are many transformation operations that take BDDs/ZDDs as inputs and output BDD/ZDD representing the resultant family after performing operations such as set union and intersection.
  However, except for some basic operations, the worst-time complexity of taking such transformation on BDDs/ZDDs has not been extensively studied, and some contradictory statements about it have arisen in the literature.
  In this paper, we show that many transformation operations on BDDs/ZDDs, including all operations for families of sets that appear in Knuth's book, cannot be performed in worst-case polynomial time in the size of input BDDs/ZDDs.
  This refutes some of the folklore circulated in past literature and resolves an open problem raised by Knuth.
  Our results are stronger in that such blow-up of computational time occurs even when the ordering, which has a significant impact on the efficiency of treating BDDs/ZDDs, is chosen arbitrarily.
\end{abstract}

\section{Introduction}
Combinatorial problems, i.e., the problems dealing with combinations of a set, frequently arise in several situations such as operations research, network analysis, and LSI design.
In solving such problems, it is often convenient to consider the set of combinations, i.e., the \emph{family of} (\emph{sub})\emph{sets}.
For example, many combinatorial optimization problems can be formulated as selecting the best combination (subset) from the family of sets satisfying constraints.
However, the number of sets in a family is possibly exponential, precluding us from explicitly retaining the family of sets.

To alleviate this issue, we can use \emph{binary decision diagram} (\emph{BDD})~\cite{bryant86} or \emph{zero-suppressed binary decision diagram} (\emph{ZDD})~\cite{minato93} that is a variant of BDD.
BDD and ZDD are data structures that compactly represent a Boolean function and a family of sets, respectively.
Since a Boolean function $f$ can be regarded as a family of sets by considering the set of assignments of input Boolean variables that evaluates $f$ to $\textit{true}$, BDD can also be regarded as a succinct representation of a family of sets.
Moreover, they support many queries about the represented family of sets, e.g., counting the number of sets and performing linear optimization over the family.
Thus, BDD and ZDD can be used as succinct indexes for a family of sets.

BDDs and ZDDs also support a number of transformation operations.
For example, when we have two BDDs representing two families of sets, we can construct a BDD representing the set union of them without extracting each set from the input families.
Using such operations, we can construct a BDD or a ZDD representing the desired family of sets.
By collecting such transformation operations, Minato~\cite{minato94} considered an algebraic system called \emph{unate cube set algebra}, whose element is a family of sets.
After that, many operations were introduced, and now the system is widely called \emph{family algebra}, whose name was given by Knuth~\cite{knuth11}.
With the algorithms performing operations on BDDs and ZDDs, every operation in the family algebra provides a useful way to construct a BDD or a ZDD representing the desired family of sets in many applications.
Many of these operations have been implemented in standard BDD and ZDD manipulation packages~\cite{Inoue2016-ae,somenzi97}, and they are used in a wide range of applications, including formal verification of circuits~\cite{Gupta2019-ap,Ito2022-fk}, analyses of power distribution networks~\cite{Inoue2015-du,Takenobu2018-yp}, and data mining~\cite{Minato2008-xa}.

However, the complexity of performing family algebra operations on BDDs and ZDDs has not been well studied, except for basic set operations.
This is because some operations require complicated recursion procedures that make complexity analysis difficult.
In particular, revealing \emph{worst-case} time complexity is important to us.
If the worst-case time complexity is large, it takes an unexpectedly long time to carry out even a \emph{single} operation for certain kinds of input.
If so, we should pay attention to the possibility of such input when we use BDDs and ZDDs as a way to implement the manipulation of families of sets.
Therefore, we investigated the worst-case time complexity of executing a single family algebra operation on BDDs and ZDDs.
Since it is known that, as described later, the sizes of a BDD and a ZDD representing the same family of sets differ in only a linear factor, this paper mainly focused on the complexity of ZDDs.
After that, we mention the complexity on BDDs.

\subsection{Related Work}
\label{ssec:related}
Since the invention of ZDD~\cite{minato93}, many family algebra operations have been proposed.
Table~\ref{tb:operations} lists basic operations.
As related work, we first describe the origins of these operations.

The first four operations in Table~\ref{tb:operations} are the most fundamental set operations set described by Minato~\cite{minato93}.
The join, quotient, and remainder operations appeared in Minato's next paper~\cite{minato94}, where the join operation is called ``product'' because a join can be considered to be the multiplication of two families when we view the union operation as an addition operation.
These operations are peculiar to the families of sets and also fundamental in defining other family algebra operations.
Later, the disjoint join and joint join operations were proposed by Kawahara et al.~\cite{kawahara16} through an extension of the join; their usage is to implicitly enumerate all of the subgraphs having a particular shape.

Restrict and permit operations were originally proposed by Coudert et al.~\cite{coudert93}, where they were called SupSet and SubSet and used for solving set cover problems or performing logic circuit minimization.
The names ``restrict'' and ``permit'' come from a study by Okuno et al.~\cite{okuno98}.
Later, nonsuperset, nonsubset, maximal, and minimal operations were introduced by Coudert~\cite{coudert97} to solve various optimization problems on graphs.
Furthermore, meet, delta, minimal hitting set, and closure operations were introduced by Knuth~\cite[\S 7.1.4 Ex.203,236,243]{knuth11} to solve various graph problems.
Table~\ref{tb:operations} contains all of the transformation operations for families of sets that appeared in Knuth's book~\cite[\S 7.1.4 Ex. 203,204,236,243]{knuth11}.

\begin{table}[tbp]
  \centering
  \caption{List of operations on family algebra.}
  \label{tb:operations}
  {\renewcommand\arraystretch{0.75}
  \footnotesize
  \setlength{\tabcolsep}{2.5pt}
  \begin{tabular}{lll}
    \toprule
    Operation & Definition & Is polytime in DD sizes? \\
    \midrule
    Union $\calF\cup\calG$ & $\{S\mid S\in\calF\vee S\in\calG\}$ & Yes~\cite{minato93}\\
    Intersection $\calF\cap\calG$ & $\{S\mid S\in\calF\wedge S\in\calG\}$ & Yes~\cite{minato93} \\
    Difference $\calF\setminus\calG$ & $\{S\mid S\in\calF\wedge S\notin\calG\}$ & Yes~\cite{minato93} \\
    Symmetric difference $\calF\oplus\calG$ & $(\calF\setminus\calG)\cup(\calG\setminus\calF)$ & Yes~\cite{minato93} \\
    \midrule
    Join $\calF\sqcup\calG$ & $\{F\cup G\mid F\in\calF, G\in\calG\}$ & \textbf{No} (Theorem~\ref{thm:join})${}^{\ast}$ \\
    Disjoint join $\calF\djoin\calG$ & $\{F\cup G\mid F\in\calF, G\in\calG, F\cap G=\emptyset\}$ & \textbf{No} (Theorem~\ref{thm:join})  \\
    Joint join $\calF\jjoin\calG$ & $\{F\cup G\mid F\in\calF, G\in\calG, F\cap G\neq\emptyset\}$ & \textbf{No} (Theorem~\ref{thm:join})  \\
    \midrule
    Meet $\calF\sqcap\calG$ & $\{F\cap G\mid F\in\calF, G\in\calG\}$ & \textbf{No} (Theorem~\ref{thm:join})${}^{\ast}$  \\
    Delta $\calF\ddelta\calG$ & $\{F\oplus G\mid F\in\calF, G\in\calG\}$ & \textbf{No} (Theorem~\ref{thm:join})${}^{\ast}$  \\
    \midrule
    Quotient $\calF\sdiv\calG$ & $\{S\mid \forall G\in\calG: S\cup G\in\calF\wedge S\cap G=\emptyset\}$ & \textbf{No} (Theorem~\ref{thm:quotient})  \\
    Remainder $\calF\modulus\calG$ & $\calF\setminus(\calG\sqcup(\calF\sdiv\calG))$ & \textbf{No} (Theorem~\ref{thm:quotient})  \\
    \midrule
    Restrict $\calF\bigtriangleup\calG$ & $\{F\in\calF \mid \exists G\in\calG:G\subseteq F\}$ & \textbf{No} (Theorem~\ref{thm:restrict})${}^{\ast}$  \\
    Permit $\calF\oslash\calG$ & $\{F\in\calF \mid \exists G\in\calG:F\subseteq G\}$ & \textbf{No} (Theorem~\ref{thm:restrict})  \\
    Nonsuperset $\calF\nonsupset\calG$ & $\{F\in\calF\mid\forall G\in\calG: G\nsubseteq F\}$ & \textbf{No} (Theorem~\ref{thm:restrict})  \\
    Nonsubset $\calF\nonsubset\calG$ & $\{F\in\calF\mid\forall G\in\calG: F\nsubseteq G\}$ & \textbf{No} (Theorem~\ref{thm:restrict})  \\
    \midrule
    Maximal $\calF^\uparrow$ & $\{F\in\calF\mid\forall F^\prime\in\calF: F\subseteq F^\prime\Rightarrow F=F^\prime\}$ & \textbf{No} (Theorem~\ref{thm:maximal})  \\
    Minimal $\calF^\downarrow$ & $\{F\in\calF\mid\forall F^\prime\in\calF: F^\prime\subseteq F\Rightarrow F=F^\prime\}$ & \textbf{No} (Theorem~\ref{thm:maximal})  \\
    \midrule
    Minimal hitting set $\calF^\sharp$ & $\{S\mid \forall F\in\calF: S\cap F\neq\emptyset\}^\downarrow$ & \textbf{No} (Theorem~\ref{thm:hitting})  \\
    Closure $\calF^\cap$ & $\{\bigcap_{S\in\calF^\prime}S\mid \calF^\prime\subseteq\calF\}$ & \textbf{No} (Theorem~\ref{thm:hitting}) \\
    \bottomrule
    \multicolumn{3}{l}{\footnotesize{${}^{\ast}$Previous studies~\cite{okuno98,knuth11} stated that they can be performed in worst-case polynomial time.}}
  \end{tabular}
  }
\end{table}

Compared to the operations themselves, the time complexity of performing them on ZDDs has not been well investigated.
Minato~\cite{minato93} proved that the first four operations in Table~\ref{tb:operations} can be performed in polynomial time with respect to the size of input ZDDs.
However, the complexity of a join operation, the most basic one among the rest, has not been fully clarified.
Knuth~\cite[\S 7.1.4 Ex. 206]{knuth11} claimed that join, as well as meet and delta, can be performed in worst-case polynomial time, but this claim lacks proof.
Conversely, Kawahara et al.~\cite{kawahara16} suggested that join, as well as disjoint join and joint join, take worst-case exponential time, again without proof.
In addition to those reports, Okuno et al.~\cite{okuno98} claimed that restrict can be performed in polynomial time, but they used the unproven proposition that join can be performed in polynomial time.
Furthermore, Knuth~\cite[\S 7.1.4 Ex. 206]{knuth11} stated that the worst-case complexity of the quotient operation was an open problem.

\subsection{Our Contribution}
\label{ssec:contribution}
In this paper, we prove that, for the operations in Table~\ref{tb:operations} aside from the first four operations, there exist polynomial-sized ZDDs such that after taking the operation, the ZDD size becomes exponential.
For example, for the join operation, we prove that there exist sequences of families of sets $\{\calF_m\}$ and $\{\calG_m\}$ such that the ZDD sizes representing $\calF_m$ and $\calG_m$ are polynomial in $m$, while the ZDD size representing $\calF_m\sqcup\calG_m$ is exponential in $m$.
This result implies that these operations cannot be performed in worst-case polynomial time with respect to the size of input ZDDs.
Thus, we refute the statement raised by Knuth~\cite{knuth11} and Okuno et al.~\cite{okuno98} that join, meet, delta, and restrict can be performed in worst-case polynomial time.
We also resolve the worst-case complexity of the quotient operation.
Moreover, we also prove that the operations in Table~\ref{tb:operations}, except for the first four operations, cannot be performed in polynomial time even when families are represented by BDDs.
Since Table~\ref{tb:operations} contains all the family algebra operations raised by Knuth~\cite{knuth11}, this paper concludes what kind of family algebra operations can be performed in polynomial time on BDDs and ZDDs.

Our result is stronger in that the resultant BDD/ZDD's size remains exponential for any \emph{order of elements}.
BDD/ZDD structures follow a total order of the elements in the base set, and it is known that this element order has a significant impact on the BDD/ZDD size.
For example, it is known that a multiplexer function can be represented in linear-sized BDD by managing the ordering while its size becomes exponential when the ordering is terrible~\cite[p.235]{knuth11}.
However, we also prove that for the sequences used in proving the above, the resultant BDD/ZDD's size is exponential in $m$ regardless of the order of elements.
This suggests that we cannot shrink the BDD/ZDD size after taking an operation by managing the element order.
Some famous BDD manipulation packages such as CUDD~\cite{somenzi97} implemented dynamic reordering, the reordering of elements after executing operations to shrink the BDD/ZDD size and thus increase the efficiency of BDD/ZDD manipulations.
Nevertheless, our results suggest that the worst-case complexity of carrying out operations cannot be polynomial, even if we employ dynamic reordering.

Note that this follows the research line of Bollig~\cite{bollig14} as follows.
Yoshinaka et al.~\cite{yoshinaka12} refuted Bryant's conjecture, which is about the complexity of performing operations on BDDs, but their counterexample was somewhat weak in that the order of elements they used was unfavorable for BDD representations.
Bollig~\cite{bollig14} later resolved this issue by proposing simpler counterexamples.
Similar to this, our results imply that the exponential blow-up in taking an operation on BDDs/ZDDs occurs not only when the order of elements is unfavorable but also when it is good for BDD/ZDD representations.

From the viewpoint of applications, BDDs/ZDDs are usually built by applying multiple family algebra operations in combination with some direct construction methods such as Simpath~\cite{knuth11} and frontier-based search~\cite{kawahara17}, which are fixed-parameter tractable algorithms with pathwidth.
However, the number of required operations stays constant in many applications.
If every operation can be performed in polynomial time, we can enjoy the polynomial time complexity in BDD/ZDD sizes even for these applications.
However, our results suggest this is not the case except for the first four operations.
In addition, although we rely on specific input examples to prove non-polynomial lower bounds, we later discuss that such blow-up may occur for other input; the detailed discussions are in Section~\ref{ssec:discussion}.
Therefore, our theoretical results have practical importance.

\section{Preliminaries}
\label{sec:prel}

\subsection{Zero-suppressed Binary Decision Diagram}
\label{ssec:zdd}
A \emph{zero-suppressed binary decision diagram} (\emph{ZDD})~\cite{minato93} is a rooted directed acyclic graph (DAG)-shaped data structure for representing a family of sets.
First, we describe the structure of ZDD.
ZDD $\ZDDZ$ consists of node set $\ZDDN$ and arc set $\ZDDA$, where the node set contains \emph{terminal} nodes $\top,\bot$ and other internal nodes.
Terminal nodes have no outgoing arcs, while every internal node has two outgoing arcs called \emph{lo-arc} and \emph{hi-arc}.
The nodes pointed by the lo-arc and the hi-arc outgoing from a node $\ZDDn$ are called \emph{lo-child} $\ZDDlo(\ZDDn)$ and \emph{hi-child} $\ZDDhi(\ZDDn)$ of $\ZDDn$.
Every internal node $\ZDDn$ is associated with an element called \emph{label} that is denoted by $\ZDDlabel(\ZDDn)$.
ZDDs must follow the \emph{ordered property}: Given a total order of elements $<$, the label of the parent node must precede that of the child node, i.e., $\ZDDlabel(\ZDDn)<\ZDDlabel(\ZDDlo(\ZDDn))$ and $\ZDDlabel(\ZDDn)<\ZDDlabel(\ZDDhi(\ZDDn))$ must hold for every internal node $\ZDDn$.
Note that the child node is always allowed to be a terminal node.
Finally, the size of a ZDD is defined by its number of nodes.

Next, we describe the semantics of ZDD.
\begin{definition}
  For ZDD node $\ZDDn$, the family $\calF_\ZDDn$ of sets represented by $\ZDDn$ is defined as follows.
  (i) If $\ZDDn=\top$, then $\calF_\ZDDn=\{\emptyset\}$.
  (ii) If $\ZDDn=\bot$, then $\calF_\ZDDn=\emptyset$.
  (iii) Otherwise, $\calF_\ZDDn=\calF_{\ZDDlo(\ZDDn)}\cup(\{\{\ZDDlabel(\ZDDn)\}\}\sqcup\calF_{\ZDDhi(\ZDDn)})$.
  Furthermore, the family of sets represented by $\ZDDZ$ is that represented by root node $\ZDDr$, where the root node is the only node having no incoming arcs.
\end{definition}
Note that $\{\emptyset\}$ and $\emptyset$ are different families; the former is the family consisting of only an empty set, while the latter is the family containing no set.
For example, Figure~\ref{fig:zdd}a is the ZDD representing the family of subsets of $\{x_1,\ldots,x_5\}$ whose cardinality is less than $3$.
Solid and dashed lines represent hi- and lo-arcs, and the element inside a circle indicates its label.

Without restrictions on the structure, there exist many ZDDs representing the same family of sets.
However, by imposing restrictions, we can obtain a \emph{canonical} ZDD, i.e., an identical ZDD structure, for every family of subsets.
This canonical form is called \emph{reduced ZDD}, and a reduced ZDD can be obtained from any ZDD by repetitively applying the following two rules.
The first rule is \emph{node sharing}: If there exist two nodes $\ZDDn$ and $\ZDDm$ whose lo-child, hi-child, and label are equal, we merge these two nodes into one (Figure~\ref{fig:zdd}b).
The second rule is \emph{zero suppression}: If there exists a node $\ZDDn$ whose hi-child is $\bot$, we eliminate $\ZDDn$ and let all of the arcs pointed to $\ZDDn$ also point to $\ZDDhi(\ZDDn)$ (Figure~\ref{fig:zdd}c).
In the reduced ZDD, no node can be eliminated by applying the above two rules.
Since applying these rules strictly decreases the size of ZDD, i.e., the number of nodes, we can deduce that the reduced ZDD of a family $\calF$ is the smallest ZDD representing $\calF$ given the total order $<$ of elements.
The size of the reduced ZDD of the family $\calF$, given the total order $<$, is denoted by $Z_<(\calF)$.
If it is clear from the context, we omit $<$ and simply write it as $Z(\calF)$.

\begin{figure}
    \centering
    \includegraphics[scale=0.8]{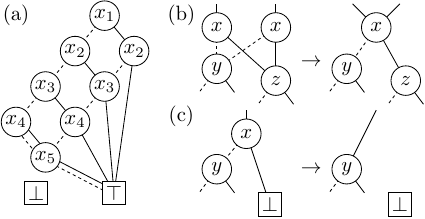}
    \caption{(a) Example of a ZDD representing the family of subsets of $\{x_1,\ldots,x_5\}$ such that the cardinality is less than $3$. (b) Schematic of node sharing. (c) Schematic of zero suppression.}
    \label{fig:zdd}
\end{figure}

We briefly compare ZDDs with BDDs.
BDD~\cite{bryant86} has the same structure (syntax) as ZDD, although its semantics is slightly different.
BDDs also follow the ordered property and have the smallest canonical form called \emph{reduced BDD}.
Given the total order $<$ of elements, the size of the reduced BDD of the family $\calF$ is denoted by $B_<(\calF)$.
The following is a famous result.
\begin{lemma}[{\cite[Eq. (126)]{knuth11}}]
  \label{lem:sizeBZ}
  For any family $\calF$ of subsets of a set of $n$ elements and any order $<$ of elements, $B_<(\calF)=\order{nZ_<(\calF)}$ and $Z_<(\calF)=\order{nB_<(\calF)}$.
\end{lemma}

\subsection{Family Algebra Operations on ZDDs}
\label{ssec:familyZDD}
In this section, we explain how the family algebra operations are performed using ZDDs and point out what makes the difference between the basic set operations (union, intersection, difference, and symmetric difference) and the other operations.

As explained in Section~\ref{ssec:zdd}, ZDD represents a family of sets in a recursive manner.
Let us consider the situation in which there are two ZDDs whose root nodes are $\ZDDn$ and $\ZDDm$ and $\ZDDlabel(\ZDDn)=\ZDDlabel(\ZDDm)=x$.
Then, the family of sets represented by them are $\calF_\ZDDn=\calF_{\ZDDlo(\ZDDn)}\cup(\{\{x\}\}\sqcup\calF_{\ZDDhi(\ZDDn)})$ and $\calF_\ZDDm=\calF_{\ZDDlo(\ZDDm)}\cup(\{\{x\}\}\sqcup\calF_{\ZDDhi(\ZDDm)})$.
The union of them is
\begin{equation}
  \calF_\ZDDn\cup\calF_\ZDDm = [\calF_{\ZDDlo(\ZDDn)}\cup\calF_{\ZDDlo(\ZDDm)}]\cup[\{\{x\}\}\sqcup(\calF_{\ZDDhi(\ZDDn)}\cup\calF_{\ZDDhi(\ZDDm)})]. \label{eq:union}
\end{equation}
This means that the ZDD representing $\calF_\ZDDn\cup\calF_\ZDDm$ can be described as follows: The root node's label is $x$, its lo-child represents $\calF_{\ZDDlo(\ZDDn)}\cup\calF_{\ZDDlo(\ZDDm)}$, and its hi-child represents $\calF_{\ZDDhi(\ZDDn)}\cup\calF_{\ZDDhi(\ZDDm)}$.
If $\ZDDlabel(\ZDDn)<\ZDDlabel(\ZDDm)$, we have a simpler recursion:
\begin{equation}
  \calF_\ZDDn\cup\calF_\ZDDm = [\calF_{\ZDDlo(\ZDDn)}\cup\calF_\ZDDm]\cup[\{\{\ZDDlabel(\ZDDn)\}\}\sqcup(\calF_{\ZDDhi(\ZDDn)}\cup\calF_\ZDDm)].
  \label{eq:union2}
\end{equation}
The case of $\ZDDlabel(\ZDDm)<\ZDDlabel(\ZDDn)$ can be handled in the same way.
By recursively expanding $\calF_\ZDDn\cup\calF_\ZDDm$ by (\ref{eq:union}) and (\ref{eq:union2}), we eventually reach terminal nodes where the union is trivial, e.g., $\calF_{\bot}\cup\calF_{\top}=\{\emptyset\}$.
Therefore, by caching the resultant ZDD nodes of $\calF_{\ZDDn^\prime}\cup\calF_{\ZDDm^\prime}$, where $\ZDDn^\prime$ and $\ZDDm^\prime$ are the child nodes of $\ZDDn$ and $\ZDDm$, respectively, we can efficiently compute the ZDD representing $\calF_\ZDDn\cup\calF_\ZDDm$.
With the cache, one can show that we can build a ZDD representing the union of two ZDDs in a time proportional to the product of input ZDD sizes.
The intersection, difference, and symmetric difference operations can be handled in almost the same way.

The other operations can also be performed in a recursive manner.
However, the recursion becomes more complicated.
Let us consider, for example, the join operation.
When $\ZDDlabel(\ZDDn)=\ZDDlabel(\ZDDm)=x$, the join becomes
\begin{equation}
  \begin{aligned}
  \calF_\ZDDn\sqcup\calF_\ZDDm = & [\calF_{\ZDDlo(\ZDDn)}\cup(\{\{x\}\}\sqcup\calF_{\ZDDhi(\ZDDn)})]\sqcup[\calF_{\ZDDlo(\ZDDm)}\cup(\{\{x\}\}\sqcup\calF_{\ZDDhi(\ZDDm)})] \\
  = & [\calF_{\ZDDlo(\ZDDn)}\sqcup\calF_{\ZDDlo(\ZDDm)}]\cup [\calF_{\ZDDlo(\ZDDn)}\sqcup(\{\{x\}\}\sqcup\calF_{\ZDDhi(\ZDDm)})]\cup \\
  &\ [(\{\{x\}\}\sqcup\calF_{\ZDDhi(\ZDDn)})\sqcup\calF_{\ZDDlo(\ZDDm)}]\cup [(\{\{x\}\}\sqcup\calF_{\ZDDhi(\ZDDn)})\sqcup(\{\{x\}\}\sqcup\calF_{\ZDDhi(\ZDDm)})] \\
  = & [\calF_{\ZDDlo(\ZDDn)}\sqcup\calF_{\ZDDlo(\ZDDm)}]\cup [\{\{x\}\}\sqcup(\calF_{\ZDDlo(\ZDDn)}\sqcup\calF_{\ZDDhi(\ZDDm)})]\cup\\
  &\ [\{\{x\}\}\sqcup(\calF_{\ZDDhi(\ZDDn)}\sqcup\calF_{\ZDDlo(\ZDDm)})]\cup [\{\{x\}\}\sqcup(\calF_{\ZDDhi(\ZDDn)}\sqcup\calF_{\ZDDhi(\ZDDm)})] \\
  = & [\calF_{\ZDDlo(\ZDDn)}\!\sqcup\!\calF_{\ZDDlo(\ZDDm)}]\!\cup\![\{\{x\}\}\sqcup((\calF_{\ZDDlo(\ZDDn)}\!\sqcup\!\calF_{\ZDDhi(\ZDDm)})\cup(\calF_{\ZDDhi(\ZDDn)}\!\sqcup\!\calF_{\ZDDlo(\ZDDm)})\cup(\calF_{\ZDDhi(\ZDDn)}\!\sqcup\!\calF_{\ZDDhi(\ZDDm)}))].
  \end{aligned}
\end{equation}
Here, the second equality holds because join distributes over the union.
This means that we should build a ZDD where the root node's lo-child represents $\calF_{\ZDDlo(\ZDDn)}\sqcup\calF_{\ZDDlo(\ZDDm)}$ and its hi-child represents $(\calF_{\ZDDlo(\ZDDn)}\sqcup\calF_{\ZDDhi(\ZDDm)})\cup(\calF_{\ZDDhi(\ZDDn)}\sqcup\calF_{\ZDDlo(\ZDDm)})\cup(\calF_{\ZDDhi(\ZDDn)}\sqcup\calF_{\ZDDhi(\ZDDm)})$.
Thus, in the recursion, we should also compute the union $\cup$ of families, which also needs a recursion like that above.
Another example is the restrict operation.
Restrict can be computed as
\begin{equation}
  \begin{aligned}
  \calF_\ZDDn\bigtriangleup\calF_\ZDDm = & [\calF_{\ZDDlo(\ZDDn)}\bigtriangleup\calF_{\ZDDlo(\ZDDm)}]\cup[\{\{x\}\}\sqcup(\calF_{\ZDDhi(\ZDDn)}\bigtriangleup(\calF_{\ZDDlo(\ZDDm)}\cup\calF_{\ZDDhi(\ZDDm)}))].
  \end{aligned}
\end{equation}
Thus, it is also necessary to compute the union of families as well as restrict.

Compared to the simple recursion for the computation of basic set operations, the complexity of such ``double recursion'' procedures are difficult to analyze.

\section{Blow-Up Operations}
\label{sec:main}
\subsection{High-Level Idea}
\label{ssec:highlevel}
As described in Section~\ref{ssec:familyZDD}, the ZDD size after performing union or intersection can be bounded by the product of the sizes of operand ZDDs, i.e., $Z(\calF\cup\calG)=\order{Z(\calF)Z(\calG)}$ and $Z(\calF\cap\calG)=\order{Z(\calF)Z(\calG)}$.
Thus, the ZDD of the union or intersection of \emph{two} ZDDs remains polynomial-sized when the operand ZDDs have polynomial size.
However, this does not hold for a non-constant number of ZDDs: even if $Z(\calF_k)=\order{\poly(m)}$ for $k=1,\ldots,m$, both $Z(\bigcup_{k=1}^{m}\calF_k)$ and $Z(\bigcap_{k=1}^{m}\calF_k)$ may become exponential in $m$.

We use such families to constitute examples of blow-up.
More specifically, for each operation, we constitute an example such that performing this operation incurs the union or intersection of multiple families.
Since we prove that the reduced ZDD representing the result of an operation will become exponential in size, we can confirm that \emph{any} algorithm for computing the resultant ZDD incurs worst-case non-polynomial complexity.
Combined with concrete instances, we prove that the worst-case complexity of family algebra operations is lower-bounded by an exponential factor.

We use the specific families of sets, hidden weighted bit function and permutation function, as explained below.
Note that they are called ``function'' because they are originally defined as a Boolean function, but we here describe them as equivalent families of sets.

\begin{definition}
  \label{def:HWB}
  A \emph{hidden weighted bit function} $\calH_m$ is a family of sets defined as $\{S\subseteq\{y_1,\ldots,y_m\}\mid y_{|S|}\in S\}$.
\end{definition}
The hidden weighted bit function $\calH_m$ can be represented as a union of elementary families.
Define $\calE_{m,k}\coloneqq \{S\subseteq\{y_1,\ldots,y_m\}\mid |S|=k, y_k\in S\}$, i.e., $\calE_{m,k}$ consists of the subsets of $\{y_1,\ldots,y_m\}$ where the cardinality is $k$ and $y_k$ is contained.
Then, $\calH_m=\bigcup_{k=1}^{m}\calE_{m,k}$.
It can be easily verified that the size of the ZDD representing $Z(\calE_{m,k})$ is $\order{m^2}$ for any order of elements (see Section~\ref{ssec:polyZDD}).
However, it is known that the ZDD representing $\calH_m$ must become exponential in size.
\begin{theorem}[\cite{bryant91}]
  For any order $<$ of elements, $B_<(\calH_m)=\morder{2^{m/5}}$. Thus, by Lemma~\ref{lem:sizeBZ}, $Z_<(\calH_m)=\morder{2^{m/5}/m}$.
\end{theorem}

\begin{definition}
  \label{def:permutation}
  A \emph{permutation function} $\calP_m$ is a family of subsets of $\{y_1,\ldots,y_{m^2}\}$ such that (i) there is exactly one element from $y_{m(i\!-\!1)\!+\!1},y_{m(i\!-\!1)\!+\!2},\ldots,y_{m(i\!-\!1)\!+\!m}$ for $i=1,\ldots,m$, and (ii) there is exactly one element from $y_{j},y_{m\!+\!j},\ldots,y_{m(m\!-\!1)\!+\!j}$ for $j=1,\ldots,m$.
\end{definition}
The permutation function $\calP_m$ is equivalent to the set of permutations:
For $S\subseteq \{y_1,\ldots,y_{m^2}\}$, we associate a binary $m\times m$ matrix where the $(i,j)$-element is $1$ if and only if $y_{m(i-1)+j}\in S$.
Then, $S\in\calP_m$ if and only if the associated matrix is a permutation matrix.

For $k=1,\ldots,m$, let $\calQ_{m,k}$ be the family of subsets of $\{y_1,\ldots,y_{m^2}\}$ such that there is exactly one element from $y_{m(k\!-\!1)\!+\!1},y_{m(k\!-\!1)\!+\!2},\ldots,y_{m(k\!-\!1)\!+\!m}$, and let $\calQ_{m,m+k}$ be those such that there is exactly one element from $y_{k},y_{m\!+\!k},\ldots,y_{m(m\!-\!1)\!+\!k}$.
Then, $\calP_m=\bigcap_{k=1}^{2m}\calQ_{m,k}$.
Here, $Z(\calQ_{m,k})=\order{m^2}$ for any order of elements, as proved in Section~\ref{ssec:polyZDD}.
However, it is again proved that the ZDD representing $\calP_m$ must become exponential in size.
\begin{theorem}[{\cite[Theorem K]{knuth11}}]
  For any order $<$ of elements, $B_<(\calP_m)=\morder{m2^m}$. Thus, by Lemma~\ref{lem:sizeBZ}, $Z_<(\calP_{m})=\morder{2^m/m}$.
\end{theorem}

We first show the exponential blow-up cases for a specific order of elements in Section~\ref{ssec:main}.
However, we see that the size of ZDD representing the hidden weighted bit function or the permutation function is exponential regardless of the order of elements.
Therefore, in Section~\ref{ssec:order}, we prove that for each family generated by the operation in Section~\ref{ssec:main}, the ZDD size representing it remains exponential regardless of the order of elements.
This means that for each operation, there exists an instance in which the input ZDD size can be polynomial by managing the element order but the output ZDD size must be exponential for any order.
Section~\ref{ssec:polyZDD} completes the proof by showing that some families can be represented by polynomial-sized ZDDs.
Finally, Section~\ref{ssec:discussion} gives some discussions on the obtained result.

\subsection{Proofs with Specific Element Order}
\label{ssec:main}

\subsubsection{Join, Disjoint Join, Joint Join, Meet, and Delta}
\label{sssec:join}
For these operations, we constitute a pair of families that incur the union of $\order{m}$ subfamilies.
Combined with $\calE_{m,k}$, the result after taking an operation contains $\bigcup_{k}\calE_{m,k}=\calH_m$, which is the hidden weighted bit function for which the ZDD size is exponential in $m$.
\begin{theorem}
  \label{thm:join}
  Let $\diamond$ be a binary operator chosen from join {\rm ($\sqcup$)}, disjoint join {\rm ($\djoin$)}, joint join {\rm ($\jjoin$)}, meet {\rm ($\sqcap$)}, and delta {\rm ($\ddelta$)}.
  Then, there exists a sequence of families $\calF_m$ and $\calG_m$ such that {\rm (i)} $\calF_m$ and $\calG_m$ are families of subsets of a set of $\order{m}$ elements, {\rm (ii)} $Z(\calF_m)+Z(\calG_m)=\order{m^3}$, and {\rm (iii)} $Z(\calF_m\diamond\calG_m)=\morder{2^{m/5}/m}$.
\end{theorem}
\begin{figure}
    \centering
    \includegraphics[scale=0.625]{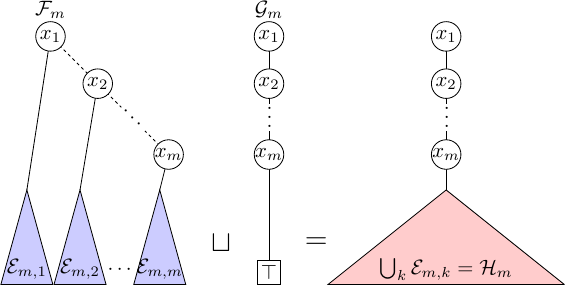}\hspace{10pt}
    \includegraphics[scale=0.625]{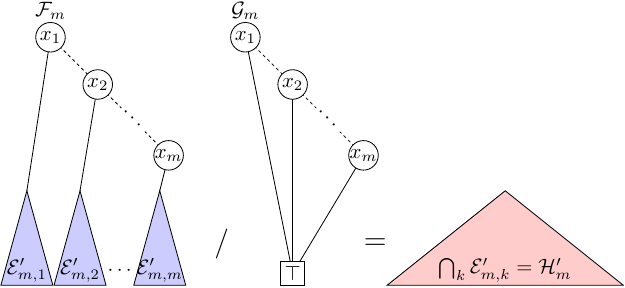}
    \caption{Example of blow-up for join (left) and quotient (right) operations. Blue triangles mean that the ZDD size representing this family is polynomial in $m$, while red triangle means that its size is exponential in $m$. Arcs going to $\bot$ terminal are omitted.}
    \label{fig:join}
\end{figure}
\begin{proof}
  Let us consider the families of subsets of $X\cup Y$, where $X\coloneqq\{x_1,\ldots,x_m\}$ and $Y\coloneqq\{y_1,\ldots,y_m\}$.
  We determine the order of elements as $x_1,\ldots,x_m,y_1,\ldots,y_m$.
  We define $\calF_m$ as
  \begin{equation*}
    \calF_m \coloneqq \bigcup_{k=1}^{m}(\{\{x_k\}\}\sqcup\calE_{m,k}).
  \end{equation*}
  Since $Z(\calE_{m,k})=\order{m^2}$ and the ZDD representing $\calF_m$ becomes the left one of Figure~\ref{fig:join} according to this order, $Z(\calF_m)=\order{m^3}$.
  
  For the join operation, we let $\calG_m\coloneqq\{X\}$, where $Z(\calG_m)=\order{m}$.
  Then,
  \begin{align*}
    \calF_m\sqcup\calG_m & \textstyle = (\bigcup_{k=1}^{m}(\{\{x_k\}\}\sqcup\calE_{m,k}))\sqcup\{X\} = \bigcup_{k=1}^{m}((\{\{x_k\}\}\sqcup\calE_{m,k})\sqcup\{X\}) \\
    & \textstyle = \bigcup_{k=1}^{m}(\{X\}\sqcup\calE_{m,k}) = \{X\}\sqcup(\bigcup_{k=1}^{m}\calE_{m,k}) = \{X\}\sqcup\calH_m,
  \end{align*}
  where the second and fourth equalities hold because join distributes over union and the third equality holds because $\{\{x_k\}\}\sqcup\{X\}=\{X\}$.
  Thus, the ZDD representing $\calF_m\sqcup\calG_m$ becomes the right one of Figure~\ref{fig:join}, meaning that the ZDD size is at least $Z(\calH_m)=\morder{2^{m/5}/m}$.
  Since every subset in $\calF_m$ has at least one element from $X$, the result of joint join $\calF_m\jjoin\calG_m$ also becomes $\{X\}\sqcup\calH_m$, leading to an exponential-sized ZDD.

  For the disjoint join operation, we let $\calG_m\coloneqq\bigcup_{k=1}^{m}\{X\setminus \{x_k\}\}$, where again $Z(\calG_m)=\order{m}$.
  Then, every subset in $\{\{x_k\}\}\sqcup\calE_{m,k}$ has intersection with all of the subsets in $\calG_m$, except for $X\setminus\{x_k\}$.
  Then,
  \begin{align*}
    \calF_m\djoin\calG_m & \textstyle = \bigcup_{k=1}^{m}((\{x_k\}\cup(X\setminus \{x_k\}))\sqcup\calE_{m,k}) = \{X\}\sqcup(\bigcup_{k=1}^{m}\calE_{m,k}) = \{X\}\sqcup\calH_m,
  \end{align*}
  meaning that $Z(\calF_m\djoin\calG_m)=\morder{2^{m/5}/m}$.

  For the meet operation, we let $\calG_m\coloneqq\{Y\}$, where $Z(\calG_m)=\order{m}$.
  Similar to join, we have $\calF_m\sqcap\calG_m=\calH_m$, meaning that $Z(\calF_m\sqcap\calG_m)=\morder{2^{m/5}/m}$.

  For the delta operation, we let $\calG_m=2^X$.
  Since $\{\{x_k\}\}\ddelta 2^X=2^X$ for any $k$, we have
  \begin{align*}
    \calF_m\ddelta\calG_m & \textstyle = \bigcup_{k=1}^{m}((\{\{x_k\}\}\ddelta 2^X)\sqcup\calE_{m,k}) = 2^X\sqcup(\bigcup_{k=1}^{m}\calE_{m,k}) = 2^X\sqcup\calH_m.
  \end{align*}
  The ZDD size of $\calF_m\ddelta\calG_m$ is at least $Z(\calH_m)=\morder{2^{m/5}/m}$.
\end{proof}

\subsubsection{Quotient and Remainder}
\label{sssec:quotient}
For the quotient operation, we constitute a pair of families such that performing an operation incurs the intersection of $\order{m}$ subfamilies.
Here, let $\calE^\comp_{m,k}\coloneqq 2^Y\setminus\calE_{m,k}$ be the complement of $\calE_{m,k}$ regarding the family of subsets of $Y$.
By De Morgan's laws, we have $\bigcap_{k}\calE^\comp_{m,k}=2^Y\setminus(\bigcup_{k}\calE_{m,k})=2^Y\setminus\calH_m\eqqcolon\calH^\comp_m$.
The ZDD size representing $\calH^\comp_m$ can be lower bounded by the following lemma.
\begin{lemma}
  \label{lem:difference}
  Suppose that two families $\calF,\calG$ of subsets of the same set satisfy $Z(\calF)=\order{f(m)}$, $Z(\calG)=\morder{g(m)}$, and $\calF\supseteq\calG$.
  Then, $Z(\calF\setminus\calG)=\morder{g(m)/f(m)}$.
\end{lemma}
\begin{proof}[Proof of Lemma~\ref{lem:difference}]
  $\calF\supseteq\calG$ implies $\calF\setminus(\calF\setminus\calG)=\calG$.
  Since the ZDD size after taking the difference can be bounded by the product of the sizes of operand ZDDs, we have $Z(\calG)=\order{Z(\calF)Z(\calF\setminus\calG)}$.
  Suppose $Z(\calF\setminus\calG)=\sorder{g(m)/f(m)}$.
  Then, $Z(\calG)=\sorder{f(m)\cdot(g(m)/f(m))}=\sorder{g(m)}$, refuting the assumption $Z(\calG)=\morder{g(m)}$.
  Therefore, $Z(\calF\setminus\calG)=\morder{g(m)/f(m)}$.
\end{proof}

Since $Z(2^Y)=\order{m}$ and $Z(\calH_m)=\morder{2^{m/5}/m}$, we have $Z(\calH^\comp_m)=\morder{2^{m/5}/m^2}$.

\begin{theorem}
  \label{thm:quotient}
  Let $\diamond$ be a binary operator chosen from quotient {\rm ($\sdiv$)} and remainder {\rm ($\modulus$)}.
  Then, there exists a sequence of families $\calF_m$ and $\calG_m$ such that {\rm (i)} $\calF_m$ and $\calG_m$ are families of subsets of a set of $\order{m}$ elements, {\rm (ii)} $Z(\calF_m)+Z(\calG_m)=\order{m^3}$, and {\rm (iii)} $Z(\calF_m\diamond\calG_m)=\morder{2^{m/5}/\poly(m)}$. 
\end{theorem}
\begin{proof}
  We again consider the families of subsets of $X\cup Y$, where $X\coloneqq\{x_1,\ldots,x_m\}$ and $Y\coloneqq\{y_1,\ldots,y_m\}$.
  We use the same order of elements: $x_1,\ldots,x_m,y_1,\ldots,y_m$.
  We define $\calF_m$ as
  \begin{equation*}
    \calF_m \coloneqq \bigcup_{k=1}^{m}(\{\{x_k\}\}\sqcup\calE^\comp_{m,k}).
  \end{equation*}
  We have $Z(\calE^\comp_{m,k})=\order{m^2}$ as proved in Section~\ref{ssec:polyZDD}, and thus $Z(\calF_m)=\order{m^3}$.
  We also define $\calG_m\coloneqq\{\{x_1\},\ldots,\{x_m\}\}$, where $Z(\calG_m)=\order{m}$.
  
  Let us consider $\calF_m\sdiv\calG_m$.
  By definition, $Y^\prime\in\calF_m\sdiv\calG_m$ if and only if $Y^\prime\subseteq Y$ and $\{x_k\}\cup Y^\prime\in\mathcal{F}_m$ for $k=1,\ldots,m$.
  From the definition of $\calF_m$, it is equivalent to $Y^\prime\in\bigcap_{k=1}^{m}\calE^\comp_{m,k}$.
  Thus, $\calF_m\sdiv\calG_m = \bigcap_{k=1}^{m}\calE^\comp_{m,k} = \calH^\comp_m$.
  This means $Z(\calF_m\sdiv\calG_m)=\morder{2^{m/5}/m^2}$.
  The ZDDs involved are depicted in Figure~\ref{fig:join}.

  For the remainder operation, we prepared the same families.
  Since $\calG_m\sqcup(\calF_m\sdiv\calG_m)=\{\{x_1\},\ldots,\{x_m\}\}\sqcup\calH^\comp_m$, $Z(\calG_m\sqcup(\calF_m\sdiv\calG_m))=\morder{2^{m/5}/m^2}$.
  Also, since $S\in\calF_m\sdiv\calG_m$ if and only if $S\cup G\in\calF_m$ for all $G\in\calG_m$, all of the subsets in $\calG_m\sqcup(\calF_m\sdiv\calG_m)$ are also contained in $\calF_m$.
  In other words, $\calF_m\supseteq\calG_m\sqcup(\calF_m\sdiv\calG_m)$.
  Therefore, by using Lemma~\ref{lem:difference}, $Z(\calF_m\modulus\calG_m)=\morder{(2^{m/5}/m^2)/m^3}=\morder{2^{m/5}/m^5}$.
\end{proof}

\subsubsection{Restrict, Permit, Nonsuperset, and Nonsubset}
\label{sssec:restrict}
These operations include inclusion relations of subsets in their definitions, which makes it difficult to generate a hidden weighted bit function as a result of the operation.
This is due to the fact that $\calH_m$ includes the universal set $Y$ as well as a singleton $\{y_1\}$.
For example, if $\calF$ is the family of subsets of $Y$ and the universal set $Y$ is included in the result of $\calF\oslash\calG$, all of the subsets in $\calF$ must be included in $\calF\oslash\calG$ due to the definition of the permit operation.

Instead, we use the permutation function.
Because every set in $\calP_m$ has cardinality $m$, the above issue can be alleviated.
More specifically, we prepared the complement of the families:
\begin{equation*}
  \calC_m \coloneqq \{S\subseteq\{y_1,\ldots,y_{m^2}\}\mid |S|=m\},\quad
  \calT_{m,k} \coloneqq \calC_m\setminus\calQ_{m,k} (= \calC_m\cap(2^Y\setminus\calQ_{m,k})).
\end{equation*}
Here, $\calC_m$ is the family of subsets with cardinality $m$, and thus $\calT_{m,k}$ also contains only the subsets with cardinality $m$.
Moreover, by De Morgan's laws,
\begin{align*}
  \bigcup_{k=1}^{2m}\calT_{m,k} & = \calC_m\cap\left(\bigcup_{k=1}^{2m}(2^Y\setminus\calQ_{m,k})\right) = \calC_m\cap\left(2^Y\setminus\left(\bigcap_{k=1}^{2m}\calQ_{m,k}\right)\right)
  = \calC_m\setminus\calP_m.
\end{align*}
We use these families $\calT_{m,k}$ to prove the following.

\begin{theorem}
  \label{thm:restrict}
  Let $\diamond$ be a binary operator chosen from restrict {\rm ($\bigtriangleup$)}, permit {\rm ($\oslash$)}, nonsuperset {\rm ($\nonsupset$)}, and nonsubset {\rm ($\nonsubset$)}.
  Then, there exists a sequence of families $\calF_m$ and $\calG_m$ such that {\rm (i)} $\calF_m$ and $\calG_m$ are families of subsets of a set of $\order{m^2}$ elements, {\rm (ii)} $Z(\calF_m)+Z(\calG_m)=\order{m^4}$, and {\rm (iii)} $Z(\calF_m\diamond\calG_m)=\morder{2^m/\poly(m)}$. 
\end{theorem}
\begin{figure}
    \centering
    \includegraphics[scale=0.625]{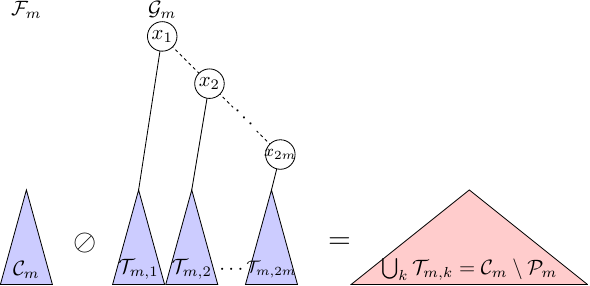}
    \includegraphics[scale=0.625]{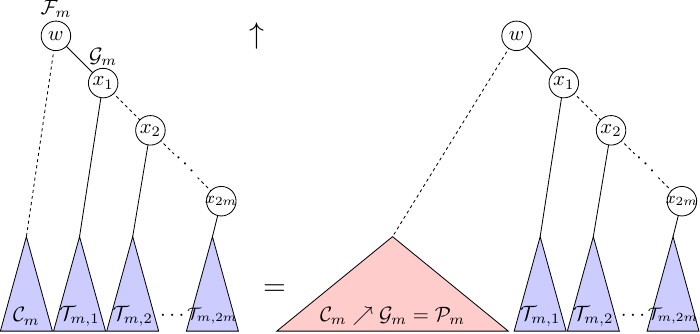}
    \caption{Example of blow-up for permit (left) and maximal (right) operations.}
    \label{fig:permit}
\end{figure}
\begin{proof}
Let us consider the families of subsets of $X\cup Y$, where $X\coloneqq\{x_1,\ldots,x_{2m}\}$ and $Y\coloneqq\{y_1,\ldots,y_{m^2}\}$.
The order of elements is  $x_1,\ldots,x_{2m}$ followed by $y_1,\ldots,y_{m^2}$.

We first consider the permit operation.
We define $\calF_m\coloneqq\calC_m$ and
\begin{equation*}
  \calG_m\coloneqq\bigcup_{k=1}^{2m}(\{\{x_k\}\}\sqcup\calT_{m,k}).
\end{equation*}
As proved in Section~\ref{ssec:polyZDD}, $Z(\calC_m)=\order{m^3}$ and $Z(\calT_{m,k})=\order{m^3}$.
Thus, $Z(\calF_m)=\order{m^3}$ and $Z(\calG_m)=\order{m^4}$.
Any set in $\calF_m=\calC_m$ consists of $m$ elements chosen from $y_1,\ldots,y_{m^2}$, and any set in $\calG_m$ consists of $m$ elements from $y_1,\ldots,y_{m^2}$ plus one element from $x_1,\ldots,x_{2m}$.
Thus, set $S\in\calF_m$ is a subset of some set in $\calG_m$ if and only if $\{x_k\}\cup S\in\calG_m$ for some $k$.
In other words, $S\in\calF_m\oslash\calG_m$ if and only if $S$ is included in $\calT_{m,k}$ for some $k$.
Since $\calC_m\supset\calT_{m,k}$ for any $k$ by definition, this means $\calF_m\oslash\calG_m=\bigcup_{k=1}^{2m}\calT_{m,k}=\calC_m\setminus\calP_m$.
Since $Z(\calC_m)=O(m^3)$ and $Z(\calP_m)=\morder{2^m/m}$, we have $Z(\calF_m\oslash\calG_m)=\morder{2^m/m^4}$ by Lemma~\ref{lem:difference}.
The ZDDs involved are depicted in Figure~\ref{fig:permit}.

The nonsubset operation can be treated with the same families.
Since $\calF_m\nonsubset\calG_m=\calF_m\setminus(\calF_m\oslash\calG_m)$ by definition, we have $\calF_m\nonsubset\calG_m=\calC_m\setminus(\calC_m\setminus\calP_m)=\calP_m$, where the last equality holds due to $\calC_m\supset\calP_m$.
Thus, $Z(\calF_m\nonsubset\calG_m)=\morder{2^m/m}$.

The restrict and nonsuperset operations can be handled by nearly the same families.
We define the same $\calG_m$ and let $\calF_m\coloneqq\{X\}\sqcup\calC_m$.
Similar to the proof of the permit operation, set $X\cup S\in\calF_m$ ($S\subseteq Y$) is a superset of some sets in $\calG_m$ if and only if $\{x_k\}\cup S\in\calG_m$ for some $k$.
This means $\calF_m\bigtriangleup\calG_m=\{X\}\sqcup(\bigcup_{k=1}^{2m}\calT_{m,k})=\{X\}\sqcup(\calC_m\setminus\calP_m)$, whose ZDD size is $\morder{2^m/m^4}$.
For the nonsuperset operation, we have $\calF_m\nonsupset\calG_m=\calF_m\setminus(\calF_m\bigtriangleup\calG_m)=\{X\}\sqcup\calP_m$, yielding $Z(\calF_m\nonsupset\calG_m)=\morder{2^m/m}$.
\end{proof}

\subsubsection{Maximal and Minimal}
\label{sssec:maximal}
For these operations, we use the close relationship with the nonsuperset and nonsubset operations.
We prepare a family having $\calF_m$ and $\calG_m$ appearing in the proof of Theorem~\ref{thm:restrict} as a subfamily.
\begin{theorem}
  \label{thm:maximal}
  Let ${}^\diamond$ be a unary operator chosen from maximal {\rm (${}^\uparrow$)} and minimal {\rm (${}^\downarrow$)}.
  Then, there exists a sequence of families $\calF_m$ such that {\rm (i)} $\calF_m$ is a family of subsets of a set of $\order{m^2}$ elements, {\rm (ii)} $Z(\calF_m)=\order{m^4}$, and {\rm (iii)} $Z(\calF_m^\diamond)=\morder{2^m/\poly(m)}$.
\end{theorem}
\begin{proof}
  Let us consider the family of subsets of $\{w\}\cup X\cup Y$, where $X\coloneqq\{x_1,\ldots,x_{2m}\}$ and $Y\coloneqq\{y_1,\ldots,y_{m^2}\}$.
  The order of elements is $w,x_1,\ldots,x_{2m}$ followed by $y_1,\ldots,y_{m^2}$.

  We first consider the maximal operation.
  We define $\calF_m$ as
  \begin{equation*}
      \calF_m\coloneqq\calC_m\cup\left[\{\{w\}\}\sqcup\calG_m\right],\quad \text{where}\ \calG_m \coloneqq \bigcup_{k=1}^{2m}(\{\{x_k\}\}\sqcup\calT_{m,k}).
  \end{equation*}
  Here, we observe that this $\calG_m$ is the same as that appearing in the proof of Theorem~\ref{thm:restrict}.
  The ZDD size is bounded as $Z(\calF_m)=\order{Z(\calC_m)+Z(\calG_m)}=\order{m^4}$.
  Every set in $\calC_m$ has $m$ elements and every set in $\{\{w\}\}\sqcup\calG_m$ has $m+2$ elements.
  Thus, every set in the latter family is maximal, while a set in the former family is maximal if and only if it is not a subset of any set included in the latter family.
  Therefore, we have
  \begin{align*}
      \calF_m^\uparrow & = \left[\calC_m\nonsubset(\{\!\{\!w\!\}\!\}\sqcup\calG_m)\right]\!\cup\!\left[\{\!\{\!w\!\}\!\}\sqcup\calG_m\right] = \left[\calC_m\nonsubset\calG_m\right]\!\cup\!\left[\{\!\{\!w\!\}\!\}\sqcup\calG_m\right] = \calP_m\cup\left[\{\!\{\!w\!\}\!\}\sqcup\calG_m\right],
  \end{align*}
  where the second equality holds because all of the sets in $\calC_m$ do not include $w$ and the last equality follows from the proof of Theorem~\ref{thm:restrict}.
  The resultant ZDD is like the right one in Figure~\ref{fig:permit}, which implies $Z(\calF_m^\uparrow)\geq Z(\calP_m)=\morder{2^m/m}$.

  The minimal can be treated in a similar way.
  We define
  \begin{equation*}
    \calF_m\coloneqq\calG_m\cup\left[\{\{w\}\}\sqcup\{\{x_1,\ldots,x_{2m}\}\}\sqcup\calC_m\right],
  \end{equation*}
  where $\calG_m$ is the same family as that above.
  We again have $Z(\calF_m)=\order{m^4}$.
  Every set in $\calG_m$ has $m+1$ elements and every set in $\{\{w\}\}\sqcup\{X\}\sqcup\calC_m$ has $3m+1$ elements.
  Thus, every set in the former family is minimal, while a set in the latter family is minimal if and only if it is not a superset of any set included in the former family.
  Now we have
  \begin{align*}
      \calF_m^\downarrow & = \calG_m\cup\left[(\{\{w\}\}\sqcup\{X\}\sqcup\calC_m)\nonsupset\calG_m\right] \\
      & = \calG_m\cup\left[\{\{w\}\}\sqcup ((\{X\}\sqcup\calC_m)\nonsupset\calG_m)\right]
      = \calG_m\cup\left[\{\{w\}\}\sqcup\{X\}\sqcup\calP_m\right],
  \end{align*}
  where the second equality holds because none of the sets in $\calG_m$ includes $w$ and the last equality follows from the proof of Theorem~\ref{thm:restrict}.
  This again implies $Z(\calF_m^\downarrow)\geq Z(\calP_m)=\morder{2^m/m}$.
\end{proof}

\subsubsection{Minimal Hitting Set and Closure}
\label{sssec:hitting}
For these operations, we can constitute much simpler examples.
\begin{theorem}
  \label{thm:hitting}
  Let ${}^\diamond$ be a unary operator chosen from minimal hitting set {\rm (${}^\sharp$)} and closure {\rm (${}^\cap$)}.
  Then, there exists a sequence of families $\calF_m$ such that {\rm (i)} $\calF_m$ is a family of subsets of a set of $\order{m^2}$ elements, {\rm (ii)} $Z(\calF_m)=\order{m^4}$, and {\rm (iii)} $Z(\calF_m^\diamond)=\morder{2^m/\poly(m)}$.
\end{theorem}
\begin{proof}
  Let $X\coloneqq\{x_1,\ldots,x_{2m}\}$ and $Y\coloneqq\{y_1,\ldots,y_{m^2}\}$.
  For $k=1,\ldots,m$, we set $S_k\coloneqq\{y_{m(k\!-\!1)\!+\!1},y_{m(k\!-\!1)\!+\!2},\ldots,y_{m(k\!-\!1)\!+\!m}\}$, and $S_{m+k}\coloneqq\{y_k,y_{m\!+\!k},\ldots,y_{m(m\!-\!1)\!+\!k}\}$.

  For minimal hitting set operation, we consider a family of subsets of $Y$.
  We define $\calF_m\coloneqq\{S_1,\ldots,S_{2m}\}$.
  Since the ZDD size can be bounded by the sum of cardinality of a set in the family~\cite{Minato2008-xa}, $Z(\calF_m)\leq\sum_i|S_i|=\order{m^2}$.
  For $S\subseteq\{y_1,\ldots,y_{m^2}\}$, we associate a binary $m\times m$ matrix, where the $(i,j)$-element is $1$ if and only if $y_{m(i-1)+j}\in S$.
  Then, $S\cap S_k\neq\emptyset$ means that the $k$-th row of the matrix has at least one $1$ and $S\cap S_{m+k}\neq\emptyset$ means that the $k$-th column of the matrix has at least one $1$.
  Thus, $S\in\calF_m^\sharp$ if and only if the corresponding matrix has at least one $1$ for any column or row and no proper subset of $S$ satisfies this property.
  The minimal matrix having this property is the permutation matrix, and thus $\calF_m^\sharp=\calP_m$, that is, the permutation function.
  This implies $Z(\calF_m^\sharp)=\morder{2^m/m}$.

  For closure operation, we consider a family of subsets of $X\cup Y$.
  For $k=1,\ldots,m$ and $\ell=1,\ldots,m$, we define $R_{k,\ell}\coloneqq (X\setminus\{x_k,x_{m\!+\!\ell}\})\cup((Y\setminus S_k\setminus S_{m\!+\!\ell})\cup\{y_{m(k\!-\!1)\!+\!\ell}\})$.
  We define $\calF_m\coloneqq\{R_{k,\ell}\mid k,\ell=1,\ldots,m\}$.
  Again, since the ZDD size can be bounded by the sum of cardinality of a set in the family~\cite{Minato2008-xa}, $Z(\calF_m)\leq\sum_{k,\ell}|R_{k,\ell}|=\order{m^4}$.
  Then, we show that $\calF_m^\cap\cap\calC_m=\calP_m$, where $\calP_m$ is the permutation function.
  If it is shown, $Z(\calF_m^\cap\cap\calC_m)=Z(\calP_m)=\morder{2^m/m}$.
  On the other hand, $Z(\calF_m^\cap\cap\calC_m)=\order{Z(\calF_m^\cap)Z(\calC_m)}$.
  Since $Z(\calC_m)=\order{m^3}$, we can deduce that $Z(\calF_m^\cap)=\morder{2^m/\poly(m)}$.

  We now prove $\calF_m^\cap\cap\calC_m=\calP_m$.
  First, we show that $\calF_m^\cap\cap\calC_m\subseteq\calP_m$.
  $R_{k,\ell}$ does not contain any element in $S_k$ and $S_{m\!+\!\ell}$ except for $y_{m(k\!-\!1)\!+\!\ell}$.
  By fixing $k$, if $\calF^\prime\subseteq\calF$ contains at least one $R_{k,\ell}$ for some $\ell$, $S=\bigcap_{S^\prime\in\calF^\prime}S^\prime$ contains at most one element from $S_k$.
  Moreover, $S$ does not contain $x_k$ if and only if $\calF^\prime$ contains at least one $R_{k,\ell}$ for some $\ell$.
  Similarly, by fixing $\ell$, if $\calF^\prime$ contains at least one $R_{k,\ell}$ for some $k$, which is equivalent to that $S$ does not contain $x_{m\!+\!\ell}$, $S$ contains at most one element from $S_{m\!+\!\ell}$.
  Now we can say that when $S$ contains no element in $X$, $S$ contains at most one element in $S_k$ for any $k=1,\ldots,2m$.
  This means that if $S$ contains no element in $X$ and $m$ elements in $Y$, $S\in\calP_m$.
  Thus, $\calF_m^\cap\cap\calC_m\subseteq\calP_m$.
  Next, we show that $\calF_m^\cap\cap\calC_m\supseteq\calP_m$.
  Let $\sigma$ be an arbitrary permutation of $1,\ldots,m$.
  Then, $\{y_{\sigma(1)},y_{m\!+\!\sigma(2)},\ldots,y_{(m-1)m\!+\!\sigma(m)}\}=R_{1,\sigma(1)}\cap R_{2,\sigma(2)}\cap\cdots\cap R_{m,\sigma(m)}$.
  This means that any set in $\calP_m$ is in $\calF_m^\cap$.
  Thus, $\calF_m^\cap\cap\calC_m\supseteq\calP_m$.
  This concludes $\calF_m^\cap\cap\calC_m=\calP_m$.
\end{proof}

\subsection{Consideration for Element Order}
\label{ssec:order}
The above proofs fix the order of elements for each operation.
Thus, there is still a possibility that the resultant ZDD size becomes smaller by managing the order of elements.
However, it seems that the size of resultant ZDD remains exponential regardless of the order of elements, since every resultant family contains a hidden weighted bit function, a permutation function, or similar families as a subfamily.
In the following, we prove that every resultant family has an exponential ZDD size regardless of the order of elements.
\begin{definition}
  Let $\calF$ be a family of subsets of set $X$, and let $Y,Y^\prime$ be the subsets of $X$ satisfying $Y\cap Y^\prime=\emptyset$.
  We define $\calF\vert_{Y,Y^\prime}$ as the family of subsets of $X\setminus(Y\cup Y^\prime)$ such that $S\in\calF\vert_{Y,Y^\prime}$ if and only if $S\cup Y\in\calF$.
\end{definition}

In other words, $\calF\vert_{Y,Y^\prime}$ is the family of sets generated from $\calF$ by first extracting the sets containing every element of $Y$, but no element of $Y^\prime$, and then eliminating all of the elements of $Y$ from every set.
This operation is called \emph{conditioning} and it is a famous result that this can be performed in polynomial time with BDDs~\cite{darwiche02}.
For the sake of completeness, we show this can also be performed in polynomial time with ZDDs, and then we prove the following.
\begin{lemma}
  \label{lem:order}
  Let $\calF$ be a family of subsets of a set $X$ of $\order{f(m)}$ elements.
  If there exist $Y,Y^\prime\subseteq X$ such that $Z_<(\calF\vert_{Y,Y^\prime})=\morder{g(m)}$ for any order $<$ of elements, we have $Z_<(\calF)=\morder{g(m)/f(m)}$ for any order $<$ of elements.
\end{lemma}

If this lemma holds, we can show that the resultant families in Section~\ref{ssec:main} all have an exponential ZDD size regardless of the order of elements.
This is because the resultant families in Section~\ref{ssec:main} all have a hidden weighted bit function, a permutation function, or its complements as a subfamily and all of them have an exponential ZDD size regardless of the order of elements; a detailed discussion is given later.
\begin{proof}[Proof of Lemma~\ref{lem:order}]
  If we can show $Z_<(\calF\vert_{Y,Y^\prime})=\order{Z_<(\calF)f(m)}$ for any $Y,Y^\prime\subseteq X$ and any order $<$ of elements, Lemma~\ref{lem:order} can be proved as follows:
  Suppose that there is an order $<$ of elements satisfying $Z_<(\calF)=\sorder{g(m)/f(m)}$.
  Then, by the above equation, we have $Z_<(\calF\vert_{Y,Y^\prime})=\sorder{(g(m)/f(m))\cdot f(m)}=\sorder{g(m)}$.
  This contradicts the assumption that $Z_<(\calF\vert_{Y,Y^\prime})=\morder{g(m)}$ for any order $<$ of elements.

  Next, we fix an arbitrary order $<$ of elements and show $Z_<(\calF\vert_{Y,Y^\prime})=\order{Z_<(\calF)f(m)}$.
  Here, we consider the operations for constructing a ZDD representing $\calF\vert_{Y,Y^\prime}$ from the ZDD of $\calF$.
  We first extract the sets that contain every element of $Y$ but do not contain any element of $Y^\prime$.
  Then, we eliminate all elements of $Y$.

  The former step can be achieved by the intersection operation.
  Let $\calG$ be the family of subsets of $X$ such that $S\in\calG$ if and only if $S$ contains all of the elements in $Y$ but does not contain any element in $Y^\prime$.
  In other words, $\calG\coloneqq\{S\subseteq X\mid S\cap Y=Y\wedge S\cap Y^\prime=\emptyset\}$.
  Then, $\calF\cap\calG$ is the desired family.
  The ZDD representing $\calG$ has the following form:
  (i) For any $x\in Y$, there is only one ZDD node labeled $x$ whose lo-child is $\bot$ while its hi-child is the next-level node.
  (ii) For any $x\in Y^\prime$, there is no node labeled $x$ by the reduction rule of ZDD.
  (iii) for any $x\in X\setminus (Y\cup Y^\prime)$, there is only one ZDD node labeled $x$ whose lo-child and hi-child are both the next-level node.
  Thus, we have $Z_<(\calG)=O(f(m))$ because the base set $X$ of $\calF$ has $\order{f(m)}$ elements and, for any element $x\in X$, there is at most one node labeled $x$.
  Finally, $Z_<(\calF\cap\calG)=O(Z_<(\calF)f(m))$.

  The latter can be achieved by eliminating the nodes labeled $x\in Y$ and replacing the branches heading it.
  For a node labeled $x\in Y$, its lo-child must be $\bot$, since the ZDD is reduced and every set in $\calF\cap\calG$ must contain $x$.
  For this node, we first make all of the arcs heading to it point to its hi-child.
  Then, we eliminate this node.
  By performing this operation for every node labeled $x\in Y$, we finally obtain the ZDD of $\calF\vert_{Y,Y^\prime}$.
  Since this operation does not increase the size of ZDD, we have $Z_<(\calF\vert_{Y,Y^\prime})=\order{Z_<(\calF)f(m)}$.
\end{proof}

Now we can show that the resultant families in the proof of Section~\ref{ssec:main} have exponential ZDD size regardless of the order of elements.
For example, for the join operation, $(\{X\}\sqcup\calH_m)\vert_{X,\emptyset}=\calH_m$ and $Z(\calH_m)=\morder{2^{m/5}/m}$ for any order of elements of $Y$ (and thus that of $X\cup Y$).
Therefore, by Lemma~\ref{lem:order}, $Z(\calF_m\sqcup\calG_m)=\morder{2^{m/5}/m^2}$ for any order of elements of $X\cup Y$.
Similar arguments hold for the other operations. 
We here show that all the resultant families in the proof of Section~\ref{ssec:main} have exponential ZDD size regardless of the order of elements.
\begin{description}
  \item[Disjoint join $\djoin$ and joint join $\jjoin$:] The resultant family of these operations in the proof of Theorem~\ref{thm:join} is $(\{X\}\sqcup\calH_m)$. Here, $(\{X\}\sqcup\calH_m)\vert_{X,\emptyset}=\calH_m$.
  \item[Meet $\sqcap$:] In the proof of Theorem~\ref{thm:join}, we already have $\calF_m\sqcap\calG_m=\calH_m$. Thus, $Z(\calF_m\sqcap\calG_m)=\morder{2^{m/5}/m}$ for any order of elements.
  \item[Delta $\ddelta$:] In the proof of Theorem~\ref{thm:join}, we have $\calF_m\ddelta\calG_m=2^X\sqcup\calH_m$. Since $(2^X\sqcup\calH_m)\vert_{X,\emptyset}=\calH_m$, $Z(\calF_m\ddelta\calG_m)=\morder{2^{m/5}/\poly(m)}$ for any order of elements.
  \item[Quotient $\sdiv$:] $Z(2^Y)=\order{m}$ and $Z(\calH_m)=\morder{2^{m/5}/m}$ for any order of elements, and $Z(\calH^\comp_m)=\morder{2^{m/5}/m^2}$ for any order of elements by Lemma~\ref{lem:difference}.
  This also holds for $\calF_m\sdiv\calG_m$ in the proof of Theorem~\ref{thm:quotient} since it equals $\calH^\prime_m$.
  \item[Remainder $\modulus$:] Since $\calF_m=\bigcup_{k}(\{\{x_k\}\}\sqcup\calE^\comp_{m,k})$ and $\calG_m\sqcup(\calF_m\sdiv\calG_m)=\{\{x_1\},\ldots,\{x_m\}\}\sqcup\calH^\comp_m$, $\calF_m\modulus\calG_m=\bigcup_{k}(\{\{x_k\}\}\sqcup(\calE^\comp_{m,k}\setminus\calH^\comp_m))$.
  Thus, $(\calF_m\modulus\calG_m)\vert_{\{x_1\},X\setminus\{x_1\}}=\calE^\comp_{m,1}\setminus\calH^\comp_m$.
  Here, $Z(\calE^\comp_{m,1})=\order{m^2}$ and $Z(\calH^\comp_m)=\morder{2^{m/5}/m^2}$ for any order of elements, and $Z(\calE^\comp_{m,1}\setminus\calH^\comp_m)=\morder{2^{m/5}/m^4}$ for any order of elements by Lemma~\ref{lem:difference}.
  Thus, by Lemma~\ref{lem:order}, $Z(\calF_m\modulus\calG_m)=\morder{2^{m/5}/m^6}$ for any order of elements because it is a family of subsets of a set with $\order{m^2}$ elements.
  \item[Permit $\oslash$ and nonsubset $\nonsubset$:] We already have $\calF_m\oslash\calG_m=\calC_m\setminus\calP_m$ and $\calF_m\nonsubset\calG_m=\calP_m$ in the proof of Theorem~\ref{thm:restrict}.
  Since $Z(\calC_m)=\order{m^3}$ and $\calP_m=\morder{2^m/m}$ for any order of elements, $Z(\calC_m\setminus\calP_m)=\morder{2^m/m^4}$ for any order of elements by Lemma~\ref{lem:difference}.
  \item[Restrict $\bigtriangleup$ and nonsuperset $\nonsupset$:] We have $(\{X\}\sqcup(\calC_m\setminus\calP_m))\vert_{X,\emptyset}=\calC_m\setminus\calP_m$ and $(\{X\}\sqcup\calP_m)\vert_{X,\emptyset}=\calP_m$; see the proof of Theorem~\ref{thm:restrict}.
  \item[Maximal ${}^\uparrow$:] In the proof of Theorem~\ref{thm:maximal}, we have $\calF_m^\uparrow\vert_{\emptyset,\{w\}\cup X}=\calP_m$.
  \item[Minimal ${}^\downarrow$:] In the proof of Theorem~\ref{thm:maximal}, we have $\calF_m^\downarrow\vert_{\{w\}\cup X,\emptyset}=\calP_m$.
  \item[Minimal hitting set ${}^\sharp$:] In the proof of Theorem~\ref{thm:hitting}, we already have $\calF_m^\sharp=\calP_m$.
  \item[Closure ${}^\cap$:] In the proof of Theorem~\ref{thm:hitting}, we have $\calF_m^\cap\cap\calC_m=\calP_m$.
  Since $Z(\calC_m)=\order{m^3}$ for any order of elements, $Z(\calF_m^\cap)=\morder{2^m/\poly(m)}$ for any order of elements.
\end{description}

\subsection{Polynomially Bounded ZDDs}
\label{ssec:polyZDD}
We complete the proof of this section by showing that the ZDD sizes of some families appearing in the previous proofs are bounded by a polynomial of $m$.
To prove the size bound, we consider the following \emph{linear network model} to distinguish whether a set is contained in the family $\calF$.
Note that the idea of a linear network model comes from Knuth's book~\cite[Theorem M]{knuth11}, where it was used to prove the bound of BDD size.
Suppose that the order of elements is $x_1<x_2<\cdots<x_n$.
There are $n$ computational modules $M_1,\ldots,M_n$.
Module $M_i$ receives an input of one bit indicating whether $x_i$ is included in the set.
Module $M_i$ sends $a_{i+1}$ bits of information to module $M_{i+1}$.
Overall, every module $M_i$ receives an input $x_i$ and $a_i$ bits of information from $M_{i-1}$ and sends $a_{i+1}$ bits of information to $M_{i+1}$.
Since module $M_1$ has no preceding module, we set $a_1=0$.
The final module, $M_n$, outputs one bit indicating whether the set is included in the family $\calF$.
An overview of the linear network model is drawn in Figure~\ref{fig:linearnet}.
The following lemma suggests that if we can construct a small linear network for the family $\calF$, the ZDD size of $\calF$ can be bounded.
\begin{lemma}
  \label{lem:linear}
  For family $\calF$ of subsets of $\{x_1,\ldots,x_n\}$, assume that we can construct the linear network model described above to distinguish whether a set is contained in $\calF$.
  Then, the size of ZDD representing $\calF$ is bounded by $Z(\calF)\leq 2+\sum_{i=1}^{n}2^{a_i}$.
\end{lemma}
\begin{proof}
  For $k=1,\ldots,n$, we consider the number of distinct subfamilies $\calF\vert_{X,Y}$, where $X\cup Y=\{x_1,\ldots,x_{k-1}\}$.
  This is because by the node sharing rule, the number of nodes labeled $x_k$ is upper-bounded by the number of possible distinct subfamilies.

  We observe that the input to module $M_i$ is $a_i$ bits.
  This means that, regardless of the inclusion of $x_1,\ldots,x_{k-1}$, the subfamily $\calF\vert_{X,Y}$ is completely determined by the information of $a_i$ bits.
  Therefore, there are at most $2^{a_i}$ distinct subfamilies, yielding the result that the number of nodes labeled $x_k$ is upper-bounded by $2^{a_i}$.
  Since there are two terminal nodes $\top$ and $\bot$, the overall ZDD size is bounded by $Z(\calF)\leq 2+\sum_{i=1}^{n}2^{a_i}$.
\end{proof}

\begin{figure}
    \centering
    \includegraphics{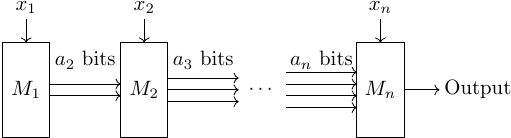}
    \caption{Schematic overview of linear network model.}
    \label{fig:linearnet}
\end{figure}

By Lemma~\ref{lem:linear}, we only have to consider a small linear network for every family.
\begin{description}
\item[$\calE_{m,k}$ in Section~\ref{ssec:highlevel}:]
The family $\calE_{m,k}$ is defined as $\{S\subseteq\{y_1,\ldots,y_m\}\mid |S|=k,y_k\in S\}$.
In judging whether $S\in\calE_{m,k}$ with a linear network, the module $M_t$ is only concerned with the number of elements from $y_1,\ldots,y_t$ in $S$ and whether $y_k$ is in $S$.
The former information can be represented with $\lceil\log (m+1)\rceil$ bits and the latter can be represented with 1 bit.
Thus, we can construct a linear network with $a_t=\lceil\log (m+1)\rceil + 1$ bits.
By Lemma~\ref{lem:linear}, we have $Z(\calE_{m,k})\leq 2+m2^{\lceil\log (m+1)\rceil+1}=\order{m^2}$.
\item[$\calQ_{m,k}$ in Section~\ref{ssec:highlevel}:]
Each of the families $\calQ_{m,k}$ $(k=1,\ldots,2m)$ is the family of subsets of $\{y_1,\ldots,y_{m^2}\}$ such that there is exactly one element from a set of $m$ selected elements.
In constructing a linear network, the module $M_t$ is only concerned with the number of selected elements in $S$: zero, one, or more than one.
This information can be represented with $2$ bits.
Thus, we have $Z(\calQ_{m,k})\leq 2+m^22^2=\order{m^2}$.
\item[$\calE_{m,k}^\comp$ in Section~\ref{sssec:quotient}:]
The linear network for $\calE_{m,k}^\comp=2^Y\setminus \calE_{m,k}$ can be the same as that for $\calE_{m,k}$, except that the output is inverted.
Thus, $Z(\calE_{m,k}^\comp)=\order{m^2}$.
\item[$\calC_m$ in Section~\ref{sssec:restrict}:]
The family $\calC_m$ is defined as $\{S\subseteq\{y_1,\ldots,y_{m^2}\}\mid |S|=m\}$. Similar to the case of $\calQ_{m,k}$, every module only retains the number of elements from $y_1,\ldots,y_t$ in $S$.
Moreover, we should only count this number until $m$; if the count exceeds $m$, we can immediately determine that $S$ is not in $\calC_m$.
This count value can be represented with $\lceil\log(m+2)\rceil$ bits.
Thus, we have $Z(\calC_m)\leq 2+m^22^{\lceil\log(m+2)\rceil}=\order{m^3}$.
\item[$\calT_{m,k}$ in Section~\ref{sssec:restrict}:]
For $\calT_{m,k}=\calC_m\setminus\calQ_{m,k}$, we can construct a linear network by combining the networks for $\calC_m$ and $\calQ_{m,k}$.
We have $\lceil\log(m+2)\rceil$ bits for $\calC_m$ and $2$ bits for $\calQ_{m,k}$.
Thus, we have $Z(\calT_{m,k})\leq 2+m^22^{\lceil\log(m+2)\rceil+2}=\order{m^3}$.
\end{description}

We finally note that the ZDD sizes of the above families remain polynomial in $m$ even if the order of elements is different from $y_1<y_2<\cdots<y_m<\cdots<y_{m^2}$.
Since the cardinality constraint is symmetric, we can reuse the same linear network for different orders of elements.
The existence of specific elements can also be treated by changing the input that is watched.

\subsection{Discussion}
\label{ssec:discussion}
Finally, we give some discussions for the presented results.
First, we argue theoretical results for BDDs.
As stated in Lemma~\ref{lem:sizeBZ}, the sizes of BDD and ZDD differ only by a linear factor of the size of the base item set.
All the results in Section~\ref{ssec:main} have the same form that the number of elements is $\order{\poly(m)}$, the input ZDD sizes are $\order{\poly(m)}$, and the output ZDD size is exponential in $m$.
Therefore, even if these families are represented by BDDs, the input BDD sizes are all $\order{\poly(m)}$, and the output BDD sizes are all exponential in $m$.
Moreover, the output BDD sizes remain exponential in $m$ for any order $<$ of elements since Lemma~\ref{lem:sizeBZ} holds for any order $<$ of elements.
This constitutes the theoretical result that the family algebra operations in Table~\ref{tb:operations}, except for the first four operations, cannot be performed in polynomial time in the input BDD sizes.

Second, we discuss how often such exponential blow-up occurs.
Although we rely on specific families, the hidden weighted bit function $\calH_m$ and the permutation function $\calP_m$, the heart of the above proofs is that even a single operation may cause us to compute the union or intersection of multiple subfamilies.
Apart from these families, it is usual that taking the union or intersection of multiple families leads to exponential blow-up.
To imagine this, we consider encoding a family described by polynomial-sized conjunctive normal form (CNF) into BDD/ZDD.
Every clause can be encoded into a polynomial-sized BDD/ZDD.
Moreover, if the entire CNF is encoded into BDD/ZDD, we can solve SAT, or even more difficult \#SAT, in linear time with respect to the size of BDD/ZDD~\cite{knuth11}.
However, it is a famous fact that SAT and \#SAT are in NP-complete and \#P-complete, respectively, meaning that they are believed not to be solved in polynomial time.
This means that for many CNFs, the BDD/ZDD after taking intersection of clauses does not remain polynomial-sized.
Therefore, apart from the specific examples used in the proof, there are many cases yielding the blow-up of BDD/ZDD size after single family algebra operation.

Finally, we argue the limitation of some of the above results that the permutation function is not such a ``devilish'' example.
The permutation function is a family of subsets of a set with $\order{m^2}$ elements and its ZDD size can only be lower bounded by $\morder{2^m/\poly(m)}$.
Since the ZDD size of the family of subsets of a set with $\order{m^2}$ can be at most $\morder{2^{m^2}/\poly(m)}$, it is far from being the worst-case.
We should investigate whether there is a family of sets generated by restrict or similar operations whose ZDD size is lower bounded by $\morder{\alpha^n/\poly(n)}$, where $\alpha>1$ and $n$ is the number of elements in the base set.

\section{Conclusion}
We proved that the worst-case complexity of carrying out certain kinds of a family algebra operation on BDDs/ZDDs once is lower bounded by an exponential factor.
These include all of the operations raised by Knuth~\cite[\S 7.1.4 Ex. 203,204,236,243]{knuth11} except for the basic set operations.
In particular, we resolved the controversy over the complexity of the join operation, which had arisen prominently in past literature.
We also resolved the open problem regarding the worst-case complexity of the quotient operation.

Future directions include the followings.
First, we only prove the lower-bound of the complexity of carrying out a single operation.
It should be investigated whether we can obtain a non-trivial upper-bound of the complexity.
Second, it is unknown whether a ``double recursion'' procedure like those in Section~\ref{ssec:familyZDD} always leads to an exponential worst-case complexity.
It is important to investigate whether there are non-trivial operations that should require a double recursion procedure even though the worst-case complexity is polynomial.



\bibliography{mybib}

\end{document}